\documentclass[journal]{IEEEtran}
\usepackage{amsmath,amssymb,amsthm,amsfonts}

\usepackage{cite}

\ifCLASSINFOpdf
\else
\fi
\usepackage{algorithmic}
\usepackage{url}

\newcommand{\fs}{\mathcal{S}}

\newcommand{\iz}{\mathbb{Z}}

\newcommand{\wtd}{\mathcal{A}}
\newcommand{\cd}{\mathcal{C}}
\newcommand{\dcd}{\mathcal{C}^{\perp}}
\newcommand{\mtd}{\mathcal{M}}
\newcommand{\mtdc}{\mathcal{M}_{\cd}}
\newcommand{\fdq}{\mathbb{F}_q}
\newcommand{\fdqn}{\mathbb{F}_{q}^n}
\newcommand{\pst}{\mathbb{P}}
\newcommand{\dpst}{\widetilde{\mathbb{P}}}
\newcommand{\pord}{\leq_{\pst}}
\newcommand{\dpord}{\leq_{\dpst}}
\newcommand{\ideal}[1]{\langle#1\rangle_{\pst}}
\newcommand{\dideal}[1]{\langle#1\rangle_{\dpst}}
\newcommand{\compl}[1]{\overline{#1}}
\newcommand{\supp}[1]{\mathrm{supp}(#1)}
\newcommand{\pwt}[1]{wt_{\pst}(#1)}
\newcommand{\dpwt}[1]{wt_{\dpst}(#1)}
\newcommand{\abs}[1]{|#1|}
\newcommand{\rk}[1]{\mathrm{rank}(#1)}
\newcommand{\dm}[1]{\mathrm{dim}(#1)}

\newtheorem{theorem}{Theorem}
\newtheorem{proposition}[theorem]{Proposition}
\newtheorem{lemma}[theorem]{Lemma}
\newtheorem{corollary}[theorem]{Corollary}

\theoremstyle{definition}
\newtheorem{definition}[theorem]{Definition}

\begin{document}
%
\title{Simple proofs for duality of generalized minimum poset weights and weight distributions of (Near-) MDS poset codes}
%
%
%

\author{Dae San Kim,~\IEEEmembership{Member,~IEEE,}
        Dong Chan Kim,
        and~Jong Yoon Hyun,
\thanks{D. S. Kim is with the Department of Mathematics, Sogang University, Seoul, Korea.
Email: dskim@sogang.ac.kr.}
\thanks{D. C. Kim is with the Department of Mathematics, Sogang University, Seoul, Korea and the Attached Institute of ETRI, P.O. Box 1, Yuseong, Daejeon, 305-600, Korea.
Email: dongchan@ensec.re.kr.}
\thanks{J. Y. Hyun is with the Department of Mathematics, Ewha Womans University, Seoul 120-750, Korea. Email: hyun33@ewha.ac.kr.}}

%
%

\markboth{Journal of \LaTeX\ Class Files,~Vol.~6, No.~1, January~2007}%
{Shell \MakeLowercase{\textit{et al.}}: Bare Demo of IEEEtran.cls for Journals}
%



\maketitle

\begin{abstract}
In 1991, Wei introduced generalized minimum Hamming weights for
linear codes and showed their monotonicity and duality. Recently,
several authors extended these results to the case of generalized
minimum poset weights by using different methods. Here, we would
like to prove the duality by using matroid theory. This gives yet
another and very simple proof of it. In particular, our argument
will make it clear that the duality follows from the well-known
relation between the rank function and the corank function of a
matroid. In addition, we derive the weight distributions of linear MDS and Near-MDS poset codes in the same spirit.
\end{abstract}


%
\IEEEpeerreviewmaketitle


\section{Introduction}

In 1991, Wei introduced the notion of generalized minimum Hamming weights for linear codes \cite{12} and showed their monotonicity and duality, motivated by its application to cryptography \cite{9}.  Actually, similar properties were considered earlier for irreducible cyclic codes by Helleseth, Kl{\o}ve and Mykkeltveit in \cite{6}.

Poset codes were first introduced in \cite{2}. They are just nonempty subsets in $\fdqn$, equipped with any poset weight instead of the usual Hamming weight. By using different methods, the duality and monotonicity results were extended to the case of generalized minimum poset weights for linear poset codes independently by Barg and Purkayastha \cite{1} and de Oliveira  Moura and Firer \cite{8}. Later, Choi and Kim \cite{3} also showed the duality for generalized minimum poset weights by exploiting yet another method.

Here, we would like to explain very briefly how the duality result is proved  in each case of \cite{1}, \cite{3}, and \cite{8}. Barg and Purkayastha in \cite{1},  as in the case of Wei's original proof in \cite{12}, do not adopt the matroid theory and exploit instead  parity check and generator matrices for linear codes. The authors in \cite{8} adopt the geometric formulation of the generalized minimum Hamming weights for projective systems in \cite{11} and use multi-set techniques, originated from \cite{5} and \cite{10}, in order to extend the proofs in [11, Theorem 4.1] to the case of generalized minimum poset weights. So their proof is far different from the original proof of Wei in \cite{12}.  Choi and Kim in \cite{3} define $P(\cd)$ and $RP(\cd)$ for linear codes $\cd$, and show the duality by using these. In doing so, they obtain more information than just the duality result.

The aim of this paper is to present simple proofs for the duality of the  generalized minimum poset weights and the weight distributions of linear MDS and Near-MDS poset codes by using only very basic facts of matroid theory \cite{13}.

In more detail, Theorem \ref{thm002} is fundamental in proving the duality in Theorem \ref{thm003} and an analogue of the corresponding Theorem 2 in \cite{12}. One remark here is that while the description involving inequality only is given in \cite{12}, that involving both inequality and equality is stated in our case(cf. (\ref{equ002}), (\ref{equ003})). We emphasize here that in showing Theorems \ref{thm002} and \ref{thm003} we only need the facts in Lemma \ref{lem001}, all of which are trivial except  perhaps (g). It is a special case of (\ref{equ001}) applied to the matroid $\mtd_\cd$ of the linear code $\cd$, and hence we may say that  the duality really follows from the well-known relation between the rank function and the corank function of a matroid. The weight distributions of linear Near-MDS poset codes were investigated in [1, Theorem 4.1] by using orthogonal array. Here we deduce them in the same spirit as showing the duality theorem. Our proof depends on the formula in (\ref{equ004}) and needs information about the values of the rank (or corank) function of the associated matroid of linear MDS and Near-MDS poset codes. For Near-MDS poset codes, we need again the relation between the rank function and the corank function of a matroid in order to have that information.


\section{Preliminaries}

The following notations will be used throughout this paper.

\begin{itemize}

\item $\fdq$ the finite field with $q$ elements

\item $[n] = \{1, \ldots, n \}$

\item For $J \subseteq [n]$, $\compl{J} = [n] \setminus J$

\item $\supp{u} = \{i : u_i \not= 0\}$, for $u=(u_1,\ldots, u_n) \in \fdqn$

\item $wt_H(u) = |\supp{u}|$ the Hamming weight of $u$

\item $\supp{D} = \cup_{u\in D} \supp{u}$, for a subset $D\subseteq \fdqn$

\item $\pst = ([n], \pord)$ a fixed poset on $[n]$

\item $\dpst = ([n], \dpord)$ the dual poset of $\pst$ on $[n]$ (i.e., $i\dpord j \Leftrightarrow j \pord i$)

\item A subset $J\subseteq [n]$ is an ideal in $\pst$ if $j \in J$ and $i \pord j$ $\Rightarrow$ $i\in J$

\item For any $J \subseteq[n]$, $\ideal{J}$ denotes the smallest ideal containing $J$ (i.e., $\ideal{J} = \{i : i \pord j, \text{ for some }j\in J\}$)

\item $\pwt{u} = \abs{\ideal{\supp{u}}}$

\item $\pwt{D} = \abs{\ideal{\supp{D}}}$

\item For $J\subseteq [n]$, $u=(u_1,\ldots, u_n) \in \fdqn$, and $D\subseteq \fdqn$,
\begin{equation*}
u|J = (u_i)_{i\in J}, \quad D|J = \{ u|J : u\in D\}
\end{equation*}

\item $\cd$ an $[n,k]$ code over $\fdq$, with a generator matrix $G$ (an $k\times n$ matrix with rank $k$) and a parity check matrix $H$ (an $(n-k)\times n$ matrix with rank $n-k$), and with $\rho$ and $\rho^\perp$ respectively the rank function and the corank function of the matriod $\mtdc$ of $\cd$. Such a $\cd$ will be viewed as a linear $\pst$-code (i.e., we regard it as a subspace of the $\pst$-space $(\fdqn, wt_{\pst})$) and the dual $\dcd$ of $\cd$ as a linear $\dpst$-code

\item $\cd_J = \cd|J$ the puncturing of $\cd$ with respect to $\compl{J}$

\item $\cd^J = \{ u|J : u\in \cd, \supp{u} \subseteq J \}$ the shortening
of $\cd$ with respect to $\compl{J}$. Hereafter we will identify $\cd^J$ with the space
\begin{equation*}l
\{ u \in \cd : \supp{u} \subseteq J\}
\end{equation*}

\item $\Phi_r(\cd)$ the set of all $r$-dimensional subspaces of $\cd$, for $0\leq r \leq \dm{\cd}$

\item $\Omega(\pst)$ the set of ideals in $\pst$

\item $\Lambda^r(\pst)$ the set of ideals in $\pst$ of size $r$

\item $S_I = \{ x \in \fdqn : \ideal{\supp{x}} = I\}$, for $I \in \Omega(\pst)$

\item $M(I)$ the set of maximal elements in $I$, for $I \in \Omega(\pst)$

\item $I_M = I \setminus M(I)$, for $I \in \Omega(\pst)$

\item $\Lambda(I) = \{J \in \Omega(\pst) : I_M \subseteq J \subseteq I\}$, for $I\in \Omega(\pst)$

\item $\{\wtd_{r,\pst}(\cd)\}_{r=0}^{n}$ the $\pst$-weight distribution of $\cd$ with $\wtd_{r,\pst}(\cd) = |\{u\in \cd : \pwt{u} = r\}|$

\end{itemize}

A matroid $\mtd$ on $\fs$ is a finite set $\fs$ together with a function(called the rank function of $\mtd$) $\rho : 2^{\fs} \rightarrow \iz_{\geq 0}$ satisfying the following three properties: for $A, B \subseteq \fs$,

\begin{itemize}
\item[(R1)] $0 \leq \rho(A) \leq |A|$,
\item[(R2)] $A \subseteq B \Rightarrow \rho(A) \leq \rho(B)$,
\item[(R3)] $\rho(A\cup B) + \rho(A \cap B) \leq \rho(A) + \rho(B)$.
\end{itemize}
A corank function $\rho^*$ is the rank function of the dual matroid $\mathcal{M}^*$ of $\mathcal{M}$. It is well-known that, for a matroid $\mathcal{M}$ with the rank function $\rho$ and
the corank function $\rho^*$, we have the following: for $A\subseteq \fs$,
\begin{equation*}
\rho^*(\mathcal{S}\setminus A) = |\mathcal{S}| - |A| - \rho(\mathcal{S}) + \rho(A),
\end{equation*}
or equivalently
\begin{equation}\label{equ001}
\rho^*(A) = |A| - \rho(\mathcal{S}) + \rho(\mathcal{S}\setminus A).
\end{equation}
For $A\subseteq[n]$, let $G|A$ and $H|A$ be respectively the submatrices of $G$ and $H$ consisting of the columns indexed by $A$. Then we observe that
\begin{equation*}
\begin{split}
\rho(A) & = \rk{G|A} = \dm{\cd|A}, \\
\rho^\perp(A) & = \rk{H|A} = \dm{\cd^\perp|A}.
\end{split}
\end{equation*}

\begin{definition} Let $\cd$ be an $[n,k]$ linear code. For $r(1\leq r \leq k)$, the $r$-th generalized minimum poset weight($\pst$-weight, if
the reference to $\pst$ is needed) is defined by
\begin{equation*}
\begin{split}
d_r^\pst(\cd) = \min \{ \pwt{D} : D\in \Phi_r(\cd)\};
\end{split}
\end{equation*}
for $s(1\leq s \leq n-k)$,
\begin{equation*}
\begin{split}
d_s^{\dpst}(\dcd) = \min \{ \dpwt{D} : D\in \Phi_s(\dcd) \}.
\end{split}
\end{equation*}
\end{definition}
The following lemma contains all the stuffs that are needed in proving Theorems 2 and 3. Here, all the statements are trivial except perhaps (g), which is just (1) applied to the matroid $\mtdc$ of $\cd$.

\begin{lemma}\label{lem001}
Let $J\subseteq[n]$. Then we have the following.
\begin{enumerate}
\item[(a)] $\supp{\cd^J} \subseteq J$.
\item[(b)] For any subset $D\subseteq \cd$, $\supp{D}\subseteq J \Leftrightarrow D\subseteq \cd^J$.
\item[(c)] If $J$ is an ideal in $\pst$, then $\ideal{J} = J$.
\item[(d)] $J$ is an ideal in $\pst$ $\Leftrightarrow$ $\compl{J}$ is an ideal in $\dpst$.
\item[(e)] $\dm{\cd^J} = \dm{\cd} - \rho(\compl{J})$.
\item[(f)] If $\supp{D}\subseteq J$, for some $D\in \Phi_r(\cd)$, then $\rho(\compl{J}) \leq \dm{\cd}-r$.
\item[(g)] $\abs{J} - \rho^\perp(J) = \dm{\cd} -\rho(\compl{J}) = \dm{\cd^J}$.
\end{enumerate}
\begin{proof}
(a), (b), (c), (d) Clear. (e) Let $\psi : \cd \rightarrow \cd|\compl{J}$ be the linear map given by $u\mapsto u|\compl{J}$. Then the kernel of this map is $\cd^J$. (f) As $D\subseteq \cd^J$ by (b), $\dm{\cd^J}\geq r$. The result now follows from (e). (g) This follows from (\ref{equ001}) and (e).
\end{proof}
\end{lemma}


\section{Proof of duality}

We do not provide the proof of the following theorem. One refers its
proof to \cite{1}.

\begin{theorem}\label{thm001}
Let $\cd$ be an $[n,k]$ linear code. Then
\begin{equation*}
\begin{split}
1 \leq d_1^{\pst}(\cd) < d_2^{\pst}(\cd) < \cdots < d_k^{\pst}(\cd) \leq n,
\end{split}
\end{equation*}
and
\begin{equation*}
\begin{split}
1 \leq d_1^{\dpst}(\dcd) < d_2^{\dpst}(\dcd) < \cdots < d_{n-k}^{\dpst}(\dcd) \leq n.
\end{split}
\end{equation*}
\end{theorem}

\begin{corollary}\label{cor001}
For $1\leq r\leq k$,
\begin{equation*}
\begin{split}
r \leq d_r^{\pst} (\cd) \leq n-k+r.
\end{split}
\end{equation*}
For $1\leq s\leq n-k$,
\begin{equation*}
\begin{split}
s \leq d_s^{\dpst} (\dcd) \leq k+s.
\end{split}
\end{equation*}
\end{corollary}

\begin{theorem}\label{thm002}
Let $\cd$ be an $[n,k]$ linear code. For $1\leq r\leq k$,
\begin{align}
d_r^{\pst} (\cd) & = \min \{ \abs{\ideal{J}} : \abs{J} - \rho^\perp(J) \geq r\} \label{equ002}\\
& = \min \{ \abs{\ideal{J}} : \abs{J} - \rho^\perp(J) = r\}. \label{equ003}
\end{align}
For $1\leq s\leq n-k$,
\begin{equation*}
\begin{split}
d_s^{\dpst} (\cd^\perp) & = \min \{ \abs{\dideal{J}} : \abs{J} - \rho^\perp(J) \geq s\} \\
& = \min \{ \abs{\dideal{J}} : \abs{J} - \rho^\perp(J) = s\}.
\end{split}
\end{equation*}
\end{theorem}
\begin{proof}
Firstly, we show that $d_{r}^{\pst} (\cd) \leq \min \{ \abs{\ideal{J}} : \abs{J} - \rho^\perp(J) = r \}$. Let $d$ denote the right hand side of
this. Let $J$ be such that $|J|-\rho^\perp(J) = r, \abs{\ideal{J}}=d$. Then, by Lemma \ref{lem001} (g), $\dm{\cd^J} = r$. So, $d_r^{\pst}(\cd) \leq \pwt{\cd^J} \leq \abs{\ideal{J}} = d$, by Lemma \ref{lem001} (a). Secondly, we show that $\min\{\abs{\ideal{J}} : \abs{J} - \rho^\perp(J) \geq r \} \leq d_r^{\pst}(\cd)$. Let $e$ denote the left hand side of this. To
show this, let $\pwt{D} = d_r^{\pst}(\cd)$, for some $D\in\Phi_r(\cd)$. Set $J=\ideal{\supp{D}}$. Then $D\subseteq \cd^J$, by Lemma \ref{lem001} (b) and $\dm{\cd^J} = \abs{J} - \rho^\perp(J) \geq r$ (cf. Lemma \ref{lem001} (g)). So, by Lemma \ref{lem001} (c), $e\leq \abs{\ideal{J}} = \abs{J} = d_r^{\pst}(\cd)$. Lastly, it is enough to see that $d\leq e$. Let $e=\abs{\ideal{J}}$, with $\abs{J} - \rho^\perp(J)\geq r$. Then we claim that $\abs{J}- \rho^\perp(J) = r$. Assume on the contrary that $\abs{J}-\rho^\perp(J) = r' > r$. Then, by the first and second steps, $d_{r'}^{\pst}(\cd) \leq \min\{\abs{\ideal{I}} : \abs{I} - \rho^\perp(I) = r'\} \leq \abs{\ideal{J}} = e \leq d_r^{\pst}(\cd)$,  a contradiction to Theorem \ref{thm001}.
\end{proof}

\begin{theorem}\label{thm003}
Let $\cd$ be an $[n,k]$ linear code and $A= \{d_r^{\pst}(\cd) : 1\leq r \leq k\}$, $B = \{ n+1-d_s^{\dpst}(\dcd) : 1\leq s \leq n-k\}$. Then $A$ and $B$ are disjoint and $[n] = A\cup B$.
\end{theorem}
\begin{proof}
It is enough to see that $A$ and $B$ are disjoint. Let $s$ be any
integer such that $1\leq s\leq n-k$. Then we need to see that $n+1-d_s^{\dpst}(\cd^\perp) \not\in A$. Firstly, let $t=k+s-d_s^{\dpst}(\cd^\perp)$ By Corollary \ref{cor001}, $s\leq d_s^{\dpst}(\cd^\perp) \leq k+s$, so that $0\leq t \leq k$. Then we claim that $d_t^{\pst}(\cd) \leq n - d_s^{\dpst}(\cd^\perp)$, so that $n+1-d_s^{\dpst}(\cd^\perp) \not= d_r^{\pst}(\cd)$, for $r\leq t$. Let $\dpwt{D} = d_s^{\dpst}(\cd^\perp)$, for some $D\in \Phi_s(\cd^\perp)$. Set $I = \dideal{\supp{D}}$.
Then $D\subseteq (\cd^\perp)^I$, by Lemma \ref{lem001} (b), and hence, by Lemma \ref{lem001} (g), $|I| - \rho(I) = \dm{(\cd^\perp)^I} \geq \dm{D} = s$. So, by Lemma \ref{lem001} (g), $\rho^\perp(\compl{I}) = |\compl{I}| - k + \rho(I) = (n-k) - (|I|-\rho(I)) \leq n-k-s$, and, with $J= \compl{I}, |J| - \rho^\perp(J) \geq (n- |I|)-(n-k-s)=k+s-d_s^{\dpst}(\cd^\perp) = t$. So, $d_t^{\pst}(\cd) \leq \abs{\ideal{J}} = |J| = n - d_s^{\dpst}(\cd^\perp)$, by Theorem \ref{thm002} and as $J$ is an ideal in $\pst$ by Lemma \ref{lem001} (d). Secondly, we must show that $n+1-d_s^{\dpst}(\cd^\perp) \not= d_{t+1}^{\pst}(\cd)$, for all $l$ with $1\leq l \leq k-t$. Suppose that $n+1-d_s^{\dpst}(\cd^\perp) = d_{t+l}^{\pst}(\cd)$, for some $l$. Then, for some $D\in \Phi_{t+l}(\cd)$, let $\pwt{D} = d_{t+l}^{\pst}(\cd) = n + 1- d_s^{\dpst}(\cd^\perp)$. If $J = \compl{I}$, with $I =\ideal{\supp{D}}$, then, by Lemma \ref{lem001} (f), $\rho(J) \leq k-t-l$. So $|J| - \rho(J) \geq (n-|I|) - (k-t-l) = (d_s^{\dpst}(\cd^\perp)-1) - (d_s^{\dpst}(\cd^\perp)-l-s) = s+l-1$. By Lemma \ref{lem001} (d), $J$ is an ideal
in $\dpst$, and, by Theorem \ref{thm002} and Lemma \ref{lem001} (c), $d_{s+l-1}^{\dpst}(\cd^\perp) \leq \abs{\dideal{J}} = \abs{J} = d_s^{\dpst}(\cd^\perp) - 1$,  which is a contradiction to Theorem \ref{thm001}.
\end{proof}


\section{Weight distributions of linear MDS and Near-MDS poset codes}

The equation (\ref{equ004}) in the following follows from [7, (3.1)], while the equation (\ref{equ005}) is clear.

\begin{proposition}
Let $I$ be an ideal in $\pst$.
\begin{itemize}
\item[(a)]
\begin{equation}\label{equ004}
\begin{split}
|\cd \cap S_I| = \sum_{J\in \Lambda(I)} (-1)^{|I|-|J|} q^{k-\rho(\compl{J})}.
\end{split}
\end{equation}
\item[(b)]
\begin{equation}\label{equ005}
\begin{split}
A_{r,\pst}(\cd) = \sum_{I\in \Lambda^r(\pst)} |\cd \cap S_I|.
\end{split}
\end{equation}
\end{itemize}
\end{proposition}

In what follows, we will denote $d_1^{\pst}(\cd)$ simply by $d$. Let $\cd$ be a MDS $\pst$-code with parameters $[n,k,d=n-k+1]$. Then one easily shows from [1, Lemma 2.2 (4)] that, for $J\in \Omega(\pst)$,
\begin{equation*}
\begin{split}
k - \rho(\compl{J}) = \left\{
                        \begin{array}{ll}
                          0, & \hbox{$|J| \leq d-1$,} \\
                          |J| - d + 1, & \hbox{$|J| \geq d$.}
                        \end{array}
                      \right.\end{split}
\end{equation*}
So, for any $I\in \Lambda^r(\pst)$, from (\ref{equ004}) we have
\begin{equation}\label{equ006}
\begin{split}
|\cd \cap S_I| & = \sum_{J \in \Lambda(I)} (-1)^{r - |J|} q^{k - \rho(\compl{J})} \\
 & = \sum_{\substack{J \in \Lambda(I) \\ |J| \leq d-1}} (-1)^{r - |J|}  + \sum_{\substack{J \in \Lambda(I) \\ |J| \geq d}} (-1)^{r - |J|} q^{|J|-d+1} \\
 & = \sum_{J \in \Lambda(I)} (-1)^{r - |J|}  + \sum_{\substack{J \in \Lambda(I) \\ |J| \geq d}} (-1)^{r - |J|} (q^{|J|-d+1} - 1).
\end{split}
\end{equation}
Now, the first sum in (\ref{equ006}) is
\begin{equation}\label{equ007}
\begin{split}
\sum_{J \in \Lambda(I)} (-1)^{r - |J|} & = \sum_{l = |I_M|}^{r} \sum_{\substack{J \in \Lambda(I) \\ |J| = l}} (-1)^{r - l}\\
&  = \sum_{l = |I_M|}^{r} (-1)^{r-l} \sum_{\substack{J \in \Lambda(I) \\ |J| = l}} 1 \\
& = \sum_{l = |I_M|}^{r} (-1)^{r-l} \binom{|M(I)|}{l-|I_M|} \\
& = \sum_{s=0}^{r-|I_M|} (-1)^{r-|I_M|-s} \binom{|M(I)|}{s} \\
& = (-1)^{|M(I)|} \sum_{s=0}^{|M(I)|} (-1)^{s} \binom{|M(I)|}{s} = 0.
\end{split}
\end{equation}
The second sum in (\ref{equ006}) is
\begin{equation}\label{equ008}
\begin{split}
\sum_{\substack{J \in \Lambda(I) \\ |J| \geq d}} & (-1)^{r - |J|}  (q^{|J|-d+1} - 1) \\
&  = \sum_{l=0}^{r-d} \sum_{\substack{J \in \Lambda(I) \\ |J| = d+l}} (-1)^{r-d-l}(q^{l+1} - 1) \\
& = \sum_{l=0}^{r-d} (-1)^{r-d-l}(q^{l+1} - 1)\sum_{\substack{J \in \Lambda(I) \\ |J| = d+l}} 1 \\
& = \sum_{l=0}^{r-d} (-1)^{r-d-l}(q^{l+1} - 1) \binom{|M(I)|}{r-d-l}\\
& = \sum_{s=0}^{r-d} (-1)^{s}(q^{r-d+1-s} - 1) \binom{|M(I)|}{s}.
\end{split}
\end{equation}
Thus we obtain the following theorem from (\ref{equ005})-(\ref{equ008}).

\begin{theorem}
Let $\cd$ be a MDS $\pst$-code with parameters $[n,k,d=n-k+1]$. Then, for $r$, with $d\leq r \leq n$,
\begin{equation*}
\begin{split}
\wtd_{r,\pst}(\cd) = \sum_{I \in \Lambda^r(\pst)} \sum_{s=0}^{r-d} (-1)^{s} \binom{|M(I)|}{s} (q^{r-d+1-s}-1).
\end{split}
\end{equation*}
\end{theorem}

Recall that an $[n,k]$ $\pst$-code is called a Near-MDS $\pst$-code if
\begin{equation*}
\begin{split}
d = d_1^\pst(\cd) = n - k, \quad d_2^\pst(\cd) = n - k + 2.
\end{split}
\end{equation*}

\begin{lemma}[\text{[1, Lemma 2.4 (1), (2)]}]\label{lem002}
The following hold.
\begin{itemize}
\item[(a)] $\cd$ is an $[n,k]$ Near-MDS $\pst$-code if and only if
\begin{itemize}
\item[(1)] Any $n-k-1$ columns of the parity check matrix $H$ are linearly independent.
\item[(2)] There exist $n-k$ linearly dependent columns of $H$.
\item[(3)] Any $n-k+1$ columns of $H$ have the full rank $n-k$.
\end{itemize}
\item[(b)] If $\cd$ is a linear Near-MDS $\pst$-code, then $\dcd$ is a linear Near-MDS $\dpst$-code.
\end{itemize}
 \end{lemma}

Now, we assume that $\cd$ is a Near-MDS $\pst$-code with parameters $[n,k,d=n-k]$. Then, from Lemma \ref{lem002} (a) above, we get
\begin{equation}\label{equ009}
\begin{split}
\rho^\perp(J) =  \left\{
                        \begin{array}{ll}
                          |J|, & \hbox{$|J| < n-k$,} \\
                          n-k, & \hbox{$|J| > n-k$.}
                        \end{array}
                      \right.
\end{split}
\end{equation}
By invoking Lemma \ref{lem001} (g) again, from (\ref{equ009}) we have, for $J\in \Omega(\pst)$,
\begin{equation}\label{equ010}
\begin{split}
k-\rho(\compl{J}) =  \left\{
                        \begin{array}{ll}
                          0, & \hbox{$|J| \leq d-1$,} \\
                          |J|-d, & \hbox{$|J| \geq d+1$.}
                        \end{array}
                      \right.
\end{split}
\end{equation}
We note here that (\ref{equ010}) also follows from (\ref{equ009}) and Lemma \ref{lem002} (b). However,  Lemma \ref{lem002} (b) is deduced in \cite{1} from the duality result in Theorem \ref{thm003}, which in turn follows from Lemma \ref{lem001} (g), as we stressed in Section I.

Then, by proceeding analogously to the MDS case, we get the following result.

\begin{theorem}[\cite{1}]\label{thm005}
Let $\cd$ be a Near-MDS $\pst$-code with parameters $[n,k,d=n-k]$. Then, for $r$, with $d\leq r \leq n$,
\begin{equation}\label{equ011}
\begin{split}
\wtd_{r,\pst}(\cd) = & \sum_{I\in \Lambda^r(\pst)} \sum_{s=0}^{r-d-1} (-1)^s \binom{|M(I)|}{s}(q^{r-d-s} -1 ) \\ & + (-1)^{r-d} \sum_{I\in \Lambda^r(\pst)} \sum_{\substack{J\in \Lambda(I) \\ |J| = d}} \wtd_J(\cd),
\end{split}
\end{equation}
where $\wtd_J(\cd) = |\cd \cap S_J|$.
\end{theorem}

In the case of Hamming weight(i.e., $wt_\pst$ with $\pst$ the antichain on $[n]$), denoting $\wtd_{r,\pst}(\cd)$ by $\wtd_r(\cd)$ as usual, we recover the following corollary in \cite{4}.

\begin{corollary}[\cite{4}]\label{cor002}
Let $\cd$ be a Near-MDS code with parameters $[n,k,d=n-k]$. Then, for $r$, with $d \leq r \leq n$,
\begin{equation*}
\begin{split}
\wtd_{r}(\cd) = & \binom{n}{r} \sum_{s=0}^{r-d-1} (-1)^s \binom{r}{s}(q^{r-d-s} -1 ) \\ & + (-1)^{r-d} \binom{n-d}{r-d} \wtd_d(\cd),
\end{split}
\end{equation*}
\begin{proof}
Now, let $\pst$ denote the antichain. Then the second double sum in (\ref{equ011}) is
\begin{equation*}
\begin{split}
\sum_{|I|=r} \sum_{\substack{  u \in \cd \\ w_H(u) = d \\ \supp{u}\subseteq I }} 1 = \binom{n-d}{r-d} \wtd_d(\cd).
\end{split}
\end{equation*}
by counting $I$, with $|I|=r$, for each fixed $u\in \cd$, with $w_H(u)=d$.
\end{proof}
\end{corollary}

\ifCLASSOPTIONcaptionsoff
  \newpage
\fi



%

%

\begin{IEEEbiographynophoto}{Dae San Kim(M'05)}
received the B.S. and M. S. degrees in mathematics
from Seoul National University, Seoul, Korea, in 1978 and 1980,
respectively, and the Ph.D. degree in mathematics from University of
Minnesota, Minneapolis, MN, in 1989. He is a professor in the
Department of Mathematics at Sogang University, Seoul, Korea. He has
been there since 1997, following a position at Seoul Women's
University. His research interests include number theory(exponential
sums, modular forms, zeta functions) and coding theory.
\end{IEEEbiographynophoto}

\begin{IEEEbiographynophoto}{Dong Chan Kim}
received the B. S. and M. S. degrees in mathematics from Sogang University, Seoul, Korea, in 2001 and 2003, respectively, and is currently pursuing a Ph. D. degree in mathematics at Sogang University. He has been working as a researcher at the Attached Institute of ETRI.
\end{IEEEbiographynophoto}


\begin{IEEEbiographynophoto}{Jong Yoon Hyun}
received the B. Sc. degree from Dongguk University in 1997 and the M. Sc. and Ph. D. degrees from POSTECH, in 2002 and 2006, respectively. He is currently with the Institute of Mathematical Sciences, Ewha Womans University, Korea. His research interests include coding theory and algebraic graph theory.
\end{IEEEbiographynophoto}





\end{document}